\documentclass[conference]{IEEEtran}

\makeatletter
 \def\blfootnote{\xdef\@thefnmark{}\@footnotetext}
 \makeatother

\usepackage{graphicx} 
\usepackage{algorithm}  
\usepackage{algorithmicx}  
\usepackage{algpseudocode}  
\usepackage{amsmath} 
\usepackage{amssymb}           
\usepackage{amsfonts}            
\usepackage{mathrsfs}            

\usepackage{color}

\usepackage{graphicx}
\usepackage{subfigure}

\usepackage[bookmarks=false, hidelinks]{hyperref}
\makeatletter
\def\UrlAlphabet{
      \do\a\do\b\do\c\do\d\do\e\do\f\do\g\do\h\do\i\do\j
      \do\k\do\l\do\m\do\n\do\o\do\p\do\q\do\r\do\s\do\t
      \do\u\do\v\do\w\do\x\do\y\do\z\do\A\do\B\do\C\do\D
      \do\E\do\F\do\G\do\H\do\I\do\J\do\K\do\L\do\M\do\N
      \do\O\do\P\do\Q\do\R\do\S\do\T\do\U\do\V\do\W\do\X
      \do\Y\do\Z}
\def\UrlDigits{\do\1\do\2\do\3\do\4\do\5\do\6\do\7\do\8\do\9\do\0}
\g@addto@macro{\UrlBreaks}{\UrlOrds}
\g@addto@macro{\UrlBreaks}{\UrlAlphabet}
\g@addto@macro{\UrlBreaks}{\UrlDigits}
\makeatother
\hypersetup{hidelinks}

\usepackage{cite}

\usepackage{verbatim}           

\usepackage{ntheorem}

\usepackage{accents}
\newcommand\munderbar[1]{\underaccent{\bar}{#1}}

\newtheorem{proposition}{Proposition}
\newtheorem*{proof}{Proof}

\newtheorem{remark}{Remark}

\makeatletter
\def\ps@IEEEtitlepagestyle{%
  \def\@oddfoot{\mycopyrightnotice}%
  \def\@oddhead{\hbox{}\@IEEEheaderstyle\leftmark\hfil\thepage}\relax
  \def\@evenhead{\@IEEEheaderstyle\thepage\hfil\leftmark\hbox{}}\relax
  \def\@evenfoot{}%
}
\def\mycopyrightnotice{%
  \begin{minipage}{\textwidth}
  \centering \scriptsize
  Copyright~\copyright~2021 IEEE. Personal use of this material is permitted. Permission from IEEE must be obtained for all other uses, in any current or future media, including\\reprinting/republishing this material for advertising or promotional purposes, creating new collective works, for resale or redistribution to servers or lists, or reuse of any copyrighted component of this work in other works by sending a request to pubs-permissions@ieee.org.
  \end{minipage}
}
\makeatother

\ifCLASSINFOpdf
  
\else

\fi

\begin{document}

\title{Spectral Graph Theory Based Resource Allocation for IRS-Assisted Multi-Hop Edge Computing}


\author{\IEEEauthorblockN{Huilian Zhang$^{\dagger}$,~Xiaofan He$^{\dagger}$,~Qingqing Wu$^{\ddagger}$,~and~Huaiyu Dai$^{*}$}
\IEEEauthorblockA{$^{\dagger}$School of Electronic Information, Wuhan University, Wuhan 430072, China, \\$^{\ddagger}$State Key Laboratory of Internet of Things for Smart City, University of Macau, Macao 999078, China, \\$^{*}$Department of Electrical and Computer Engineering, North Carolina State University, Raleigh, NC 27606, USA}
\IEEEauthorblockA{E-mails: \{huilianzhang, xiaofanhe\}@whu.edu.cn,~qingqingwu@um.edu.mo,~hdai@ncsu.edu}
}
 
\maketitle

\begin{abstract}
The performance of mobile edge computing (MEC) depends critically on the quality of the wireless channels. From this viewpoint, the recently advocated intelligent reflecting surface (IRS) technique that can proactively reconfigure wireless channels is anticipated to bring unprecedented performance gain to MEC. In this paper, the problem of network throughput optimization of an IRS-assisted multi-hop MEC network is investigated, in which the phase-shifts of the IRS and the resource allocation
of the relays need to be jointly optimized. However, due to the coupling among the transmission links of different hops caused by the utilization of the IRS and the complicated multi-hop network topology, it is difficult to solve the considered problem by directly applying existing optimization techniques. Fortunately, by exploiting the underlying structure of the network topology and spectral graph theory, it is shown that the network throughput can be well approximated by the second smallest eigenvalue of the network Laplacian matrix. This key finding allows us to develop an effective iterative algorithm for solving the considered problem. Numerical simulations are performed to corroborate the effectiveness of the proposed scheme.
\end{abstract}

\IEEEpeerreviewmaketitle


{\blfootnote{This work was supported in part by the NSFC Grant No. 61901305, the Wuhan University Start-up Grant No. 1501–600460001, the grants SRG2020-00024-IOTSC and FDCT 0108/2020/A, as well as the NSF grants ECCS-1444009 and CNS-1824518.}}

\section{Introduction}

\noindent To support the various computation-intensive services in the new information era, mobile-edge computing (MEC) \cite{mao2017survey} has emerged as a promising computing paradigm. The core idea of MEC is to push the computing resource to the network edge and allow mobile devices (MDs) to offload their computation tasks to the nearby edge servers for further processing \cite{he2019peace, wang2016mobile, mahmoodi2019optimal}. As a result, its performance depends critically on the quality of the communication links between the MDs and edge servers.

As a recently emerging communication technology, an intelligent reflecting surface (IRS) \cite{wu2019towards, wu2019intelligent, wu2019beamforming, zhou2020robust, di2020hybrid, li2020reconfigurable} is composed of a large number of passive reflection elements, and each of the elements can induce phase-shift of the incident signals. By coordinating the reflections of all elements, the reflected signals can add constructively at a desired point, thus improving the transmission link quality. Driven by these appealing features, there is a recent surge of interests in studying IRS-assisted MEC systems \cite{bai2020latency, liu2020intelligent, bai2020resource, cao2019intelligent} where an IRS is deployed to assist task offloading by mitigating the signal propagation induced impairments and improving the quality of the corresponding transmission links. For example, in \cite{bai2020latency} and \cite{liu2020intelligent}, IRS is leveraged to assist the task offloading from single-antanna MDs to an edge server co-located with a multi-antenna access point (AP). In \cite{bai2020resource},  IRS is employed in a wireless-powered MEC network to improve the links both for wireless energy transfer and task offloading. In addition, IRS is adopted to improve the wireless channel quality of a millimeter-wave-MEC system in \cite{cao2019intelligent} for a lower task offloading latency.

Nevertheless, the research on IRS-assisted MEC is still in its infancy, and the existing relevant works are mainly focused on single-hop MEC. In many scenarios, multi-hop MEC is needed to achieve more satisfactory performance. For example, in a disaster site or urban areas with poor cellular coverage, users may need to request computing services from a remote edge server through a multi-hop ralay network. However, multi-hop MEC is more sensitive to propagation induced impairments as compared to the single-hop case. To the best of our knowledge, how to leverage IRS to improve the performance of multi-hop MEC still remains largely unexplored. 

Motivated by the above, an IRS-assisted multi-hop MEC network is studied in this paper. Specifically, to maximize the network throughput, the phase-shifts of the IRS as well as the power and bandwidth allocation of the relays in the edge network need to be jointly optimized. However, this is highly non-trivial. In particular, the utilization of the IRS brings coupling among the transmission links of different hops. Besides, the complicated multi-hop network topology makes it more difficult to derive a closed-form expression of the network throughput, which thus prevents us from directly applying existing optimization techniques to solve the considered problem.

To overcome this technical challenge, the network throughput optimization of the considered IRS-assisted multi-hop MEC network is converted into a max-flow problem in a
directed-graph. By exploiting the underlying structure of the network topology, it is shown that the max-flow in this directed graph is equivalent to that of its undirected counterpart. Based on prior results from spectral graph theory \cite{chung1997spectral,ford1956maximal,bhattacharya2010graph,he2014dynamic}, a direct consequence of this key finding is that the network throughput can be well approximated by the second smallest eigenvalue of the network Laplacian matrix, whose gradient can be readily computed. This in turn allows us to effectively increase the network eigenvalue (and thus approximately the throughput) by adjusting the IRS phase-shifts as well as the power and bandwidth allocation via the generic gradient descent method. To the best of our knowledge, this work is among the first to investigate IRS-assisted multi-hop MEC and analyze its throughput from the perspective of spectral graph theory.


\section{System Model and Problem Formulation}
\noindent In this section, the system model will be presented first, followed by the problem formulation. 

\begin{figure} [!thb]
\centering
\vspace{-.3cm}
\includegraphics[width=0.8\columnwidth]{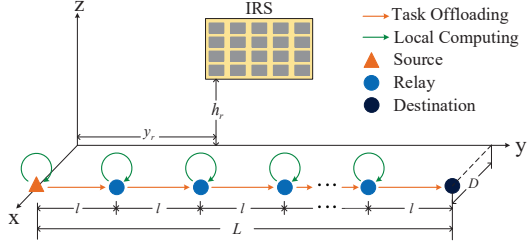}
\vspace{-.3cm}
\caption{An IRS-assisted multi-hop MEC system.}
\label{fig_sim}
\vspace{-.6cm}
\end{figure}

\subsection{System Model}
\noindent Consider an IRS-assisted multi-hop MEC system, as depicted in Fig. 1. In this scenario, a terrestrial MD (i.e., the source node) intends to offload its tasks to a remote edge server (i.e., the destination node). In addition, it is assumed that there are multiple MDs in between and they can form a multi-hop edge relay network. For the ease of description, the source node, the relay nodes, and the destination node are numbered from $1$ to $N$; define $\mathcal N\buildrel \Delta \over =\{1,2,...,N-1\}$. In contrast to the conventional relays \cite{ju2012full} that merely perform forwarding, the edge relays (i.e., the intermediate MDs) considered in this work can also help execute part of their received computation tasks by local computing. Specifically, node-$i$ ($i\in \mathcal N$) can locally compute a portion of the tasks and offload the rest to node-$(i+1)$. To facilitate task offloading, an off-the-shelf IRS with $M$ reflection elements is deployed in the environment (e.g., on the facade of nearby building). 

Denote by $L$ the source-to-destination (s-d) distance. As depicted in Fig. 1, the source node, the relay nodes, and the destination node are assumed to be located in a line. The locations of the source and the destination nodes are denoted by $(D,0,0)$ and $(D,L,0)$, respectively. The IRS is assumed to be a uniform rectangular array (URA) with $M_z$ rows and $M_y$ columns, and is located in the y-z plane. Without loss of generality, the bottom-left reflection element of the IRS is taken as the reference point to represent the location of the IRS, which is assumed to be $h_r$ meters (m) above the ground. The horizontal coordinate of node-$i$ ($i\in \mathcal N$) is given by ${{\bf{q}}_i} = {[D,(i - 1) \cdot l]^{\rm T}}$, with $l = L/(N - 1)$. The horizontal coordinate of the IRS is ${{\bf{q}}_r} = {[0,{y_r}]^{\rm T}}$, with $0 < {y_r} < L$.

\subsubsection{Task Offloading Model}
Assisted by the IRS, the channel ${H_i} \in \mathbb{C}$ from node-$i$ to node-$(i+1)$ is given by
\begin{equation}
{H_i} = {\bf{h}}_{r,i + 1}^H\boldsymbol{\Theta} {{\bf{h}}_{i,r}} + {h_{i,i + 1}},
\end{equation}
with ${h_{i,i + 1}} \in \mathbb{C}$ denoting the coefficient of the direct channel from node-$i$ to node-$(i+1)$, ${{\bf{h}}_{i,r}} \in \mathbb{C}{^{M \times 1}}$ denoting the channel vector from node-$i$ to the IRS, and ${\bf{h}}_{r,i + 1}^H \in \mathbb{C}{^{1 \times M}}$ being the channel vector from the IRS to node-$(i+1)$. The diagonal matrix $\boldsymbol{\Theta} {\rm{ = diag(}}{e^{j{\theta _1}}}{\rm{,}}{e^{j{\theta _2}}}{\rm{,}} \ldots {\rm{,}}{e^{j{\theta _M}}}{\rm{)}}$ is the phase-shift matrix of the IRS, with ${\theta _m \in [0, 2\pi)}$ being the phase-shift of the $m$-th ($m\in \mathcal M$, $\mathcal M\buildrel \Delta \over =\{1,2,...,M\}$) element of the IRS, and $j \buildrel \Delta \over = \sqrt { - 1}$ being the imaginary unit. To avoid inter-channel interference, frequency-division multiple access is adopted for task offloading. Denote by $B$ the total available system bandwidth, and $\eta_i\in(0, 1)$ the fraction of bandwidth allocated to the link from node-$i$ to node-$(i+1)$. Then, it follows that $\sum\nolimits_{i = 1}^{N - 1} {{\eta _i}}  = 1$. The transmission data rate $R_i^o$ for task offloading from node-$i$ to node-$(i+1)$ is given by
\begin{equation}
R_i^o = {\eta _i}B{\log _2}\left( {1 + ({\mu _i}{P_i}{{\left| {{H_i}} \right|}^2})/({\eta _i}B{N_0})} \right),
\end{equation}
with $N_0$ denoting the noise power spectral density, $P_i$ denoting the total power for task offloading and local computing of node-$i$, and $\mu_i$ ($0 < {\mu _i} < 1$) denoting the fraction of power allocated for task offloading. Therefore, the power/bandwidth allocation vectors of all the MDs can be compactly written as ${\boldsymbol{\mu }} \buildrel \Delta \over = {[{\mu _1},{\mu _2}, \ldots ,{\mu _{N - 1}}]^{\rm T}}$ and ${\boldsymbol{\eta }} \buildrel \Delta \over = {[{\eta _1},{\eta _2}, \ldots ,{\eta _{N - 1}}]^{\rm T}}$, respectively.

\subsubsection{Channel Model}
The wireless links between the IRS and the MDs are assumed to be dominated by LoS components due to the flexible IRS deployment. Rician fading is assumed for the direct link between eack pair of MDs, since there are usually more scatters and obstacles near the ground \cite{wu2019intelligent}. Denote the distances between node-$i$ and the IRS and that between the IRS and node-$(i+1)$ by $d_{i,r}$ and $d_{r,i}$, respectively. The channel vectors from node-$i$ to the IRS and that from the IRS to node-$(i+1)$ are given respectively by ${{\bf{h}}_{i,r}} = \sqrt {\rho {{\left( {{{{d_{i,r}}}}/{{{d_0}}}} \right)}^{ - {\alpha _1}}}} {{\bf{\bar h}}_{i,r}}$ and ${\bf{h}}_{r,i}^H = \sqrt {\rho {{\left( {{{{d_{r,i}}}}/{{{d_0}}}} \right)}^{ - {\alpha _1}}}} {\bf{\bar h}}_{r,i}^H$, with $\rho$ the path-loss at the reference distance ${d_0} = 1$ (m) and $\alpha_1$ the path-loss exponent of the LoS links. Here, ${{\bf{\bar h}}_{i,r}} \in \mathbb{C}{^{M \times 1}}$ and ${\bf{\bar h}}_{r,i}^H \in \mathbb{C}{^{1 \times M}}$ represent respectively the corresponding normalized LoS paths \cite{wu2019intelligent}. The channel coefficient from node-$i$ to node-$(i+1)$ is given by
\begin{equation}
{h_{i,i + 1}} = \chi  \cdot \left( {\sqrt {\frac{\beta }{{1 + \beta }}} h_{i,i + 1}^{{\rm{LoS}}} + \sqrt {\frac{1}{{1 + \beta }}} h_{i,i + 1}^{{\rm{NLoS}}}} \right),
\end{equation}
with $\chi  = \sqrt {\rho {{\left( {{\textstyle{{{d_{i,i + 1}}} / {{d_0}}}}} \right)}^{ - {\alpha _2}}}}$ capturing the path-loss effect. The notations ${\alpha _2}$ and $\beta$ denote the corresponding path loss exponent and the Rician factor, respectively, and $d_{i,i+1}$ is the distance between this pair of nodes. ${h_{i,i + 1}^{{\rm{LoS}}}}$ and ${h_{i,i + 1}^{{\rm{NLoS}}}}$ denote respectively the normalized LoS component with unit modulus and the circularly-symmetric-complex-Gaussian distributed (with zero mean and unit variance) normalized NLoS (Rayleigh fading) component.

\subsubsection{Local Computing Model}
Based on the above discussions, the power for local computing at node-$i$ is given by $P_i^l = (1 - {\mu _i}){P_i}$. It follows that $P_i^l = \kappa f_i^3$, with $f_i$ the corresponding CPU frequency, and $\kappa$ a coefficient depending on the chip architecture \cite{he2019peace,zhang2013energy} . Accordingly, it has ${f_i} = \sqrt[3]{{P_i^l/\kappa }}$. Assume that $\varepsilon$ CPU cycles are required to process one bit of task locally \cite{he2019peace}, and hence the data rate of local computing at node-$i$ is given by
\begin{equation}
R_i^l = \sqrt[3]{{P_i^l/\kappa }}/\varepsilon  = \sqrt[3]{{(1-\mu_i){P_i}/\kappa }}/\varepsilon. 
\end{equation}

\subsection{Problem Formulation}
\noindent To maximize the throughput (i.e., the number of task processed in unit time) of the IRS-assisted multi-hop MEC network, a max-flow problem \cite{ford1956maximal} is formulated as follows. In particular, the task-handling procedure of the considered IRS-assisted multi-hop MEC network can be represented by a directed graph $G_d$ as depicted in Fig. 2. In $G_d$, the solid lines with arrows represent single-hop task-offloading, while the dotted lines with arrows represent local computing. The associated edge weights $R^o_i$ and $R^l_i$ are determined by the task offloading rate and the local computing rate specified in (2) and (4), respectively. The rationale behind the virtual edges represented by the dotted lines is as follows: When a task is locally computed by a certain node-$i$, it is equivalent to, in terms of network throughput, the case that this task is sent to the server with negligible transmission delay and then processed there.\footnote{Similar to existing literature \cite{zhang2017energy,wang2017computation}, it is assumed in this work that the time of sending the computation results back to the source is negligible.}

\begin{figure} [!htb]
\vspace{-.4cm}
\centering
\includegraphics[width=0.8\columnwidth]{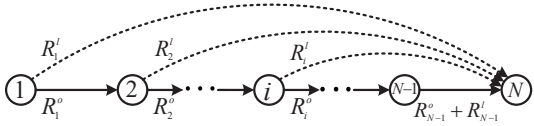}
\vspace{-.3cm}
\caption{The directed graph $G_d$ for the considered multi-hop MEC network.}
\label{fig_sim}
\vspace{-.2cm}
\end{figure}

From the above description, it is not difficult to verify that the throughput of this IRS-assisted multi-hop MEC network is essentially the max-flow $F_{s,d}$ from the source node (i.e., node-$1$) to the destination node (i.e., node-$N$) in the directed graph $G_d$. Hence, the objective of this work is to maximize the max-flow $F_{s,d}$ by jointly optimizing the phase-shifts ${\boldsymbol{\Theta }}$ of the IRS, the power allocation ${\boldsymbol{\mu }}$ of the MDs, and the system bandwidth allocation ${\boldsymbol{\eta }}$. Mathematically, the problem can be formulated as

\begin{subequations}
\begin{align}
\text{(P1):}{\kern 6pt}&\mathop {\max }\limits_{{\boldsymbol{\Theta }},{\boldsymbol{\mu }},{\boldsymbol{\eta }}} {F_{s,d}}\\
\text{s.t.}{\kern 5pt}&0 \le {\theta _m} \le 2\pi ,{\kern 15pt} \forall m\in \mathcal M,\\
&0 \leq {\mu _i} , {\eta _i} \leq 1, {\kern 12pt} \forall i \in \mathcal N,\\
&\sum\nolimits_{i = 1}^{N - 1} {{\eta _i}}  = 1.
\end{align}
\end{subequations}
One of the main difficulties in solving the above problem is that there exists no closed-form expression of  $F_{s,d}$, which makes the conventional optimization methods inapplicable. To overcome this difficulty, a spectral graph theory based optimization method will be proposed in the next section. 

\section{The Proposed Scheme}
\noindent In this section, a spectral graph theory based method is developed to solve the optimization problem (P1). The basic idea of the proposed solution is to first use an appropriate spectral graph quantity to bound the objective function in (P1), and then use the gradient descent method to gradually optimize this spectral graph quantity, hoping to achieve a reasonably good solution of (P1). 

In particular, according to prior results on spectral graph theory \cite{he2014dynamic}, the max-flow of an undirected graph can be properly approximated by the so-called weighted Cheeger's constant. However, the graph $G_d$ in the considered problem is a directed one. To this end, consider another graph $\tilde G$ that is identical to $G_d$, except that all the edges in $\tilde G$ are {\it{undirected}}. As shown in Proposition 1, due to its special topological structure, the max-flow $F_{s,d}$ in $G_d$ is equal to the max-flow ${\tilde F_{s,d}}$ in $\tilde G$.

\begin{proposition}
For any directed graph $G_d$ with the structure as shown in Fig. 2, its max-flow from the source to the destination equals that of the undirected graph $\tilde G$.
\end{proposition}

\begin{proof}
Please see Appendix A in\cite{theproof}. $\hfill\blacksquare$ 
\end{proof}

\begin{remark}
\textup{The main merit of this result is to allow us to improve the max-flow $F_{s,d}$ (or, equivalently, the network throughput) in the directed graph $G_d$ by using spectral graph theoretic methods, as elaborated below.}
\end{remark}

\subsection{Preliminary on the Weighted Cheeger's Constant}
\noindent For the undirected graph $\tilde G$, its adjacency matrix \cite{bhattacharya2010graph} is defined as ${\bf{A}} \buildrel \Delta \over = [{a_{ij}}]_{i,j = 1}^N$, where
\begin{equation}
{a_{ij}} \buildrel \Delta \over = \left\{ \begin{array}{l}
R_i^o,{\kern 1pt} \forall i\in \mathcal N\setminus \{N-1\},{\kern 1pt} j = i + 1{\kern 1pt} ,\\
R_i^l,{\kern 2pt}  \forall i\in \mathcal N\setminus \{N-1\},{\kern 1pt} j = N{\kern 1pt} ,\\
R_{N - 1}^o + R_{N - 1}^l,{\kern 1pt} i = N - 1,{\kern 1pt} j = N{\kern 1pt} ,\\
0,{\kern 1pt} {\rm{otherwise}}{\kern 1pt} .
\end{array} \right.
\end{equation}
The elements in the lower triangular part of ${\bf{A}}$ can be similarly obtained by interchanging the subscripts $i$ and $j$. Given the adjacency matrix ${\bf{A}}$, the max-flow ${\tilde F_{s,d}}$ of $\tilde G$ is given by \cite{ford1956maximal}
\begin{equation}
{\tilde F_{s,d}}  = \mathop {\min }\limits_{\left\{ {S|{s} \in S,{d} \in \bar S} \right\}} \sum\nolimits_{i \in S,j \in \bar{S}} {{a_{ij}}},
\end{equation}
where $S$ is a subset of nodes in $\tilde G$ and $\bar{S}$ denotes its complement. 
Due to the combinational nature of (7), it is quite difficult to directly optimize ${\tilde F_{s,d}}$ (by adjusting the phase-shifts ${\boldsymbol{\Theta }}$ of the IRS, the power allocation ${\boldsymbol{\mu }}$ of the MDs, and the system bandwidth allocation ${\boldsymbol{\eta }}$). 

Prior results from spectral graph theory \cite{chung1997spectral,he2014dynamic} reveal that by properly assigning large weights to the source node $s$ and the destination node $d$, the corresponding weighted Cheeger's constant $\mathscr{C}$ can be used as a good estimate of $\tilde F_{s,d}$. Specifically, $\mathscr{C}$ is defined as:
\begin{equation}
\mathscr{C} = \mathop {\min }\limits_S \frac{{\sum\nolimits_{i \in S,j \in \bar S} {{a_{ij}}} }}{{\min \left\{ {{{\left| S \right|}_W},{{\left| {\bar S} \right|}_W}} \right\}}},
\end{equation}
where ${\left| S \right|_W} = \sum\nolimits_{i \in S} {{w_i}}$ is the weighted cardinality, with ${w_i} \ge 0$ the weight assigned to node-$i$. In addition, the following weighted Cheeger's inequality holds \cite{he2014dynamic}
\begin{equation}
\lambda /2 \le \mathscr{C} \le \sqrt {2{\delta _{\max }}\lambda /{w_{\min }}},
\end{equation}
where ${w_{\min }} \buildrel \Delta \over = {\min _i}{w_i}$. The second smallest eigenvalue $\lambda$ of the weighted Laplacian matrix ${{\bf{L}}_W}$ is defined as 
\begin{equation}
\lambda \buildrel \Delta \over = \mathop {\inf }\limits_{\boldsymbol{g} \bot {{\bf{W}}^{1/2}}{\bf{1}}} \frac{{{\boldsymbol{g}^{\rm T}}{{\bf{L}}_W}\boldsymbol{g}}}{{{\boldsymbol{g}^{\rm T}}\boldsymbol{g}}},
\end{equation}
where ${{\bf{L}}_W} \buildrel \Delta \over = {{\bf{W}}^{ - 1/2}}{\bf{L}}{{\bf{W}}^{ - 1/2}}$, with the diagonal matrix ${\bf{W}} \buildrel \Delta \over = {\rm{diag}}\left\{ {{w_1}, \ldots ,{w_N}} \right\}$; ${\bf{L}} \buildrel \Delta \over = {\bf{D}} - {\bf{A}}$ is the Laplacian matrix of the undirected graph with ${\bf{D}} \buildrel \Delta \over = {\rm{diag}}\left\{ {{\delta _1}, \ldots ,{\delta _N}} \right\}$ the generalized degree matrix and ${\delta _i} \buildrel \Delta \over = \sum\nolimits_{\left\{ {j|j \ne i} \right\}} {{a_{ij}}}$.

\begin{remark}
\textup{The above results reveal that ${\tilde F_{s,d}}$ can be approximated by $\mathscr{C}$ and both the upper and the lower bounds of $\mathscr{C}$ are increasing in $\lambda$. Consequently, the max-flow ${\tilde F_{s,d}}$ in $\tilde G$ can be improved by increasing the corresponding second smallest eigenvalue $\lambda$ of the weighted Laplacian matrix. }
\end{remark}

\subsection{Joint Phase-Shifts, Power Allocation, and Bandwidth Allocation Optimization}

\noindent Based on the above discussions, an effective algorithm is developed in the following to obtain a good $\tilde{F}_{s,d}$ by increasing $\lambda $ through jointly optimizing the phase-shifts ${\boldsymbol{\Theta }}$ of the IRS, the power allocation ${\boldsymbol{\mu }}$ of the MDs, and the system bandwidth allocation ${\boldsymbol{\eta }}$.

 Let ${\boldsymbol{\theta }} = {\rm{diag\{ }}{\bf{\Theta }}\}$ be a column vector consisting of the elements on the main diagonal of the maxtrix ${\bf{\Theta }}$. It follows from (1), (2), (4), (6), and (10) that $\lambda$ is a function of the variables ${\boldsymbol{\Theta }}$, ${\boldsymbol{\mu }}$, and ${\boldsymbol{\eta }}$ , and for the ease of presentation, write $\lambda = \lambda({\bf{x}})$, where ${\bf{x}} = {[{{\boldsymbol{\theta }}^{\rm T}},{{\boldsymbol{\mu }}^{\rm T}},{{\boldsymbol{\eta }}^{\rm T}}]^{\rm T}}$. According to \cite{he2014dynamic}, if $\lambda$ is differentiable with respect to (w.r.t.) the variables ${\boldsymbol{\Theta }}$, ${\boldsymbol{\mu }}$, and ${\boldsymbol{\eta }}$, $\lambda$ can be improved by properly increasing the variables along the gradient direction $\nabla \lambda({\bf{x}}) = {\rm{\{ }}\partial \lambda/\partial {x_k}{\rm{\} }}_{k = 1}^{M + 2(N - 1)}$; here, $\{x_k\}_{k=1}^{M}$, $\{x_k\}_{k=M+1}^{M+N-1}$, and $\{x_k\}_{k=M+N}^{M+2(N-1)}$ correspond to ${\boldsymbol{\theta }}$, ${\boldsymbol{\mu }}$, and ${\boldsymbol{\eta }}$, respectively. In the following, the computation of $\nabla \lambda({\bf{x}})$ is given in detail. Specifically, the partial derivative of $\partial \lambda/\partial {x_k}$ is given by \cite{he2014dynamic}
\begin{equation}
\frac{{\partial \lambda}}{{\partial {x_k}}} = {{\bf{v}}^{\rm T}}\frac{{\partial {{\bf{L}}_W}}}{{\partial {x_k}}}{\bf{v}} = {\sum\limits_{i,j} {\left( {\frac{{{v_i}}}{{\sqrt {{w_i}} }} \!-\! \frac{{{v_j}}}{{\sqrt {{w_j}} }}} \right)} ^2}\frac{{\partial {a_{ij}}}}{{\partial {x_k}}},
\end{equation}
where $\bf{v}$ is the eigenvector of the matrix ${{\bf{L}}_W}$ corresponding to $\lambda$, ${v_i}$ (${v_j}$) is the $i$-th ($j$-th) component of $\bf{v}$. Therefore, to obtain the spatial gradient direction of $\lambda$, one need to compute $\{ \partial {a_{ij}}/\partial {x_k}\} _{k = 1}^{M + 2(N - 1)}$ in advance. By observing Fig. 2, it is not difficult to see that this boils down to compute the gradient of the offloading rate $\{ \partial R_i^o/\partial {x_k}\} _{k = 1}^{M + 2(N - 1)}$ and that of the local computing rate $\{ \partial R_i^l/\partial {x_k}\} _{k = 1}^{M + 2(N - 1)}$.

Note that the direct channel cofficient from node-$i$ to node-$(i+1)$ is given by ${h_{i,i + 1}} = \xi {e^{j\omega }}$, with $\xi$ and $\omega$ the amplitude and the phase of the complex element ${h_{i,i + 1}}$, respectively. Similarly, it has ${{\bf{h}}_{i,r}} = {[{\gamma _1}{e^{j{\varphi _1}}},{\gamma _2}{e^{j{\varphi _2}}}, \ldots ,{\gamma _M}{e^{j{\varphi _M}}}]^{\rm T}}$ and ${\bf{h}}_{r,i + 1}^H{\rm{ = [}}{\delta _1}{e^{j{\psi _1}}}{\rm{,}}{\delta _2}{e^{j{\psi _2}}}{\rm{,}} \ldots {\rm{,}}{\delta _M}{e^{j{\psi _M}}}{\rm{]}}$. Therefore, the channel ${H_i} \in \mathbb{C}$ from node-$i$ to node-$(i+1)$ can be written as
\begin{equation}
\begin{array}{l}
\!{H_i} = \sum\nolimits_{m = 1}^M {{\delta _m}{\gamma _m}{e^{j({\theta _m} + {\psi _m} + {\varphi _m})}} + \xi {e^{j\omega }}} \\
 {\kern 10pt}= \underbrace {\sum\nolimits_{m = 1}^M {{\delta _m}{\gamma _m}\sin ({\theta _m} + {\psi _m} + {\varphi _m}) + \xi \sin \omega } }_{{I_i}}\\
 {\kern 10pt}+ j\underbrace {( {\sum\nolimits_{m = 1}^M {{\delta _m}{\gamma _m}\cos ({\theta _m} + {\psi _m} + {\varphi _m})+\xi \cos \omega } } )}_{{Q_i}}
\end{array}.
\end{equation}

According to (2) and (12), the partial derivative of ${R_i^o}$ w.r.t. the phase-shift $\theta_m$ of the $m$-th reflection element is given by
\begin{equation}
\frac{{\partial R_i^o}}{{\partial {\theta _m}}} = \frac{{{\eta _i}B{\mu _i}{P_i}}}{{({\eta _i}B{N_0} + {\mu _i}{P_i}{{\left| {{H_i}} \right|}^2}) \cdot \ln 2}} \cdot \frac{{\partial {{\left| {{H_i}} \right|}^2}}}{{\partial {\theta _m}}},
\end{equation}
where
\begin{equation}
\begin{array}{l}
\frac{{\partial {{\left| {{H_i}} \right|}^2}}}{{\partial {\theta _m}}} = 2{I_i}{\delta _m}{\gamma _m}\cos ({\theta _m} + {\psi _m} + {\varphi _m})\\
{\kern 65pt} - 2{Q_i}{\delta _m}{\gamma _m}\sin ({\theta _m} + {\psi _m} + {\varphi _m}).
\end{array}
\end{equation}
Similarly, the other relevant partial derivatives are given by
\begin{equation}
\frac{{\partial R_i^o}}{{\partial {\mu _i}}} = \frac{{{\eta _i}B{P_i}{{\left| {{H_i}} \right|}^2}}}{{({\eta _i}B{N_0} + {\mu _i}{P_i}{{\left| {{H_i}} \right|}^2}) \cdot \ln 2}},
\end{equation}
\begin{equation}
\begin{array}{l}
\frac{{\partial R_i^o}}{{\partial {\eta _i}}} = B{\log _2}\left( {1 + \frac{{{\mu _i}{P_i}{{\left| {{H_i}} \right|}^2}}}{{{\eta _i}B{N_0}}}} \right)\\
{\kern 100pt} - \frac{{B{\mu _i}{P_i}{{\left| {{H_i}} \right|}^2}}}{{({\eta _i}B{N_0} + {\mu _i}{P_i}{{\left| {{H_i}} \right|}^2}) \cdot \ln 2}},
\end{array}
\end{equation}
and
\begin{equation}
\partial R_i^l/\partial {\mu _i} =  - {P_i}{\left( {(1 - {\mu _i}){P_i}/\kappa } \right)^{ - 2/3}}/(3\kappa \varepsilon ).
\end{equation}
Moreover, according to the local computing model specified in Section II-A, it has $\partial R_i^l/\partial {\theta _m} = 0$ and $\partial R_i^l/\partial {\eta _i} = 0$.

After obtaining the set of gradients ${\rm{\{ }}\partial \lambda/\partial {x_k}{\rm{\} }}_{k = 1}^{M + 2(N - 1)}$ by using (11)-(17), the value of $\lambda$ can be improved by adjusting ${\boldsymbol{\Theta }}$, ${\boldsymbol{\mu }}$, and ${\boldsymbol{\eta}}$ along their gradient directions $\nabla \lambda({\bf{x}}) = {\rm{\{ }}\partial \lambda/\partial {x_k}{\rm{\} }}_{k = 1}^{M + 2(N - 1)}$. Specifically, according to the first order Taylor expansion
\begin{equation}
\lambda({{\bf{x}}^0} + \Delta {\bf{x}}) = \lambda({{\bf{x}}^0}) + \nabla \lambda{({{\bf{x}}^0})^{\rm T}}\Delta {\bf{x}} + o(||\Delta {\bf{x}}||),
\end{equation}
for a given point ${{\bf{x}}^0}$, the value of $\lambda({{\bf{x}}^0} + \Delta {\bf{x}})$ can be improved by maximizing $\nabla \lambda{({{\bf{x}}_0})^{\rm T}}\Delta {\bf{x}}$. This indicates that $\lambda$ can be optimized by solving the following linear programming
\begin{subequations}
\begin{align}
\text{(P2):}{\kern 5pt}&\mathop {\max }\limits_{\Delta {\boldsymbol{\Theta }},\Delta {\boldsymbol{\mu }},\Delta {\boldsymbol{\eta }}} \sum\limits_{m = 1}^M {\frac{{\partial \lambda}}{{\partial {\theta _m}}}}  \cdot \Delta {\theta _m} + \sum\limits_{i = 1}^{N - 1} {\frac{{\partial \lambda}}{{\partial {\mu _i}}} \cdot \Delta {\mu _i} } \nonumber\\
&+\sum\limits_{i = 1}^{N - 1} {\frac{{\partial \lambda}}{{\partial {\eta _i}}} \cdot \Delta {\eta _i}} ,\\
\text{s.t.}{\kern 5pt}&0 \le {\theta _m} \!+\! \Delta {\theta _m} \le 2\pi ,\left| \Delta {\theta _m}\right| \le\Delta {\bar\theta _m}, \forall m \in \mathcal M,\\
&0 \le {\mu _i} + \Delta {\mu _i} \le 1,\left| \Delta {\mu _i}\right| \le\Delta {\bar\mu _i}, {\kern 15pt}\forall i \in \mathcal N,\\
&0 \le {\eta _i} + \Delta {\eta _i} \le 1,\left|\Delta {\eta _i}\right| \le\Delta {\bar\eta _i}, {\kern 19pt}\forall i \in \mathcal N,\\
&\sum\nolimits_{i = 1}^{N - 1} {\Delta {\eta _i} = 0},
\end{align}
\end{subequations}
where $\Delta {\boldsymbol{\Theta }}$, $\Delta {\boldsymbol{\mu }}$, and $\Delta {\boldsymbol{\eta }}$ denote the adjustment vectors of the phase-shifts of the IRS, the power allocation of the MDs, and the system bandwidth allocation, respectively. Constraint (19e) ensures that the total bandwidth remains unchanged. $\Delta {\boldsymbol{\bar \Theta }} = \{ \Delta {\bar \theta _m}\} _{m = 1}^M$, $\Delta {\boldsymbol{\bar \mu }} = \{ \Delta {\bar \mu _i}\} _{i = 1}^{N-1}$, and $\Delta {\boldsymbol{\bar \eta }} = \{ \Delta {\bar \eta _i}\} _{i = 1}^{N-1}$ are defined as the upper bounds of adjustments to smooth the updating procedure.

Based on the above description, a joint phase-shifts, power allocation, and bandwidth allocation optimization (JPPBO) algorithm is proposed as summarized in Algorithm 1.\footnote{In Agorithm 1, $\Delta {\boldsymbol{\munderbar\Theta }}$, $\Delta {\boldsymbol{\munderbar\mu}}$, and $\Delta {\boldsymbol{\munderbar\eta }}$ are the predefined lower bounds of $\Delta {\boldsymbol{\bar\Theta }}$, $\Delta {\boldsymbol{\bar\mu }}$, and $\Delta {\boldsymbol{\bar\eta }}$, respectively, and $\tau$ ($0<\tau<1$) is a predefined scaling factor.}

\begin{algorithm} [!htb]  
    \caption{The proposed JPPBO algorithm}  
   \begin{algorithmic}[1] 
       \State {\bf{Input:}} ${\boldsymbol{\Theta }},{\boldsymbol{\mu }},{\boldsymbol{\eta }}$, $\Delta {\boldsymbol{\bar\Theta }}$, $\Delta {\boldsymbol{\bar\mu }}$, $\Delta {\boldsymbol{\bar\eta }}$, $\Delta {\boldsymbol{\munderbar\Theta }}$, $\Delta {\boldsymbol{\munderbar\mu}}$, $\Delta {\boldsymbol{\munderbar\eta }}$.
       \State Compute the initial value of ${\tilde F_{s,d}^0}$.
       \For {$t=1,2,\ldots,T_{max}$}
             \State Compute the gradient ${\rm{\{ }}\partial \lambda/\partial {x_k}{\rm{\} }}_{k = 1}^{M + 2(N - 1)}$.
             \State Compute $\Delta {\boldsymbol{\Theta }}$, $\Delta {\boldsymbol{\mu }}$, and $\Delta {\boldsymbol{\eta }}$ by solving (19).
             \State Generate temporary variables: ${\boldsymbol{\hat \Theta }}\leftarrow {\boldsymbol{\Theta }}{\rm{ + }}\Delta {\boldsymbol{\Theta }}$, \par ${\boldsymbol{\hat \mu }} \leftarrow {\boldsymbol{\mu }}{\rm{ + }}\Delta {\boldsymbol{\mu }}$, and ${\boldsymbol{\hat\eta }} \leftarrow {\boldsymbol{\eta }}{\rm{ + }}\Delta {\boldsymbol{\eta }}$.
             \State Precompute the magnitude of max-flow ${\tilde F_{s,d}^{pre}}$.
              \If {${\tilde F_{s,d}^{pre}} \ge {\tilde F_{s,d}^{t-1}}$}
                  \State Update the variables: ${\boldsymbol{\Theta }} \leftarrow {\boldsymbol{\hat \Theta }}$, ${\boldsymbol{\mu }} \leftarrow {\boldsymbol{\hat \mu }}$, ${\boldsymbol{\eta }} \leftarrow {\boldsymbol{\hat \eta }}$.
                  \State Update ${\tilde F_{s,d}^t} \leftarrow  {\tilde F_{s,d}^{pre}}$, and record ${\tilde F_{s,d}^t}$.
               \Else
                  \State Adjust the step size: $\Delta {\boldsymbol{\bar \Theta }} \leftarrow \max \{ \tau \Delta {\boldsymbol{\bar \Theta }},\Delta {\boldsymbol{\munderbar\Theta }}\}$, \par \quad $\Delta {\boldsymbol{\bar \mu }} \!\leftarrow\! \max \{ \tau \Delta {\boldsymbol{\bar \mu }},\Delta {\boldsymbol{\munderbar\mu }}\}$, $\Delta {\boldsymbol{\bar \eta }} \!\leftarrow\! \max \{ \tau \Delta {\boldsymbol{\bar \eta }},\Delta {\boldsymbol{\munderbar\eta }}\}$.
                  \State Go to step 5.
              \EndIf
       \EndFor
       \State {\bf{Output:}} ${\boldsymbol{\Theta }},{\boldsymbol{\mu }},{\boldsymbol{\eta }}$.
    \end{algorithmic}  
\end{algorithm}  

\begin{figure*} [t!]
\begin{minipage}[t]{0.33\textwidth}
\includegraphics[scale = 0.25]{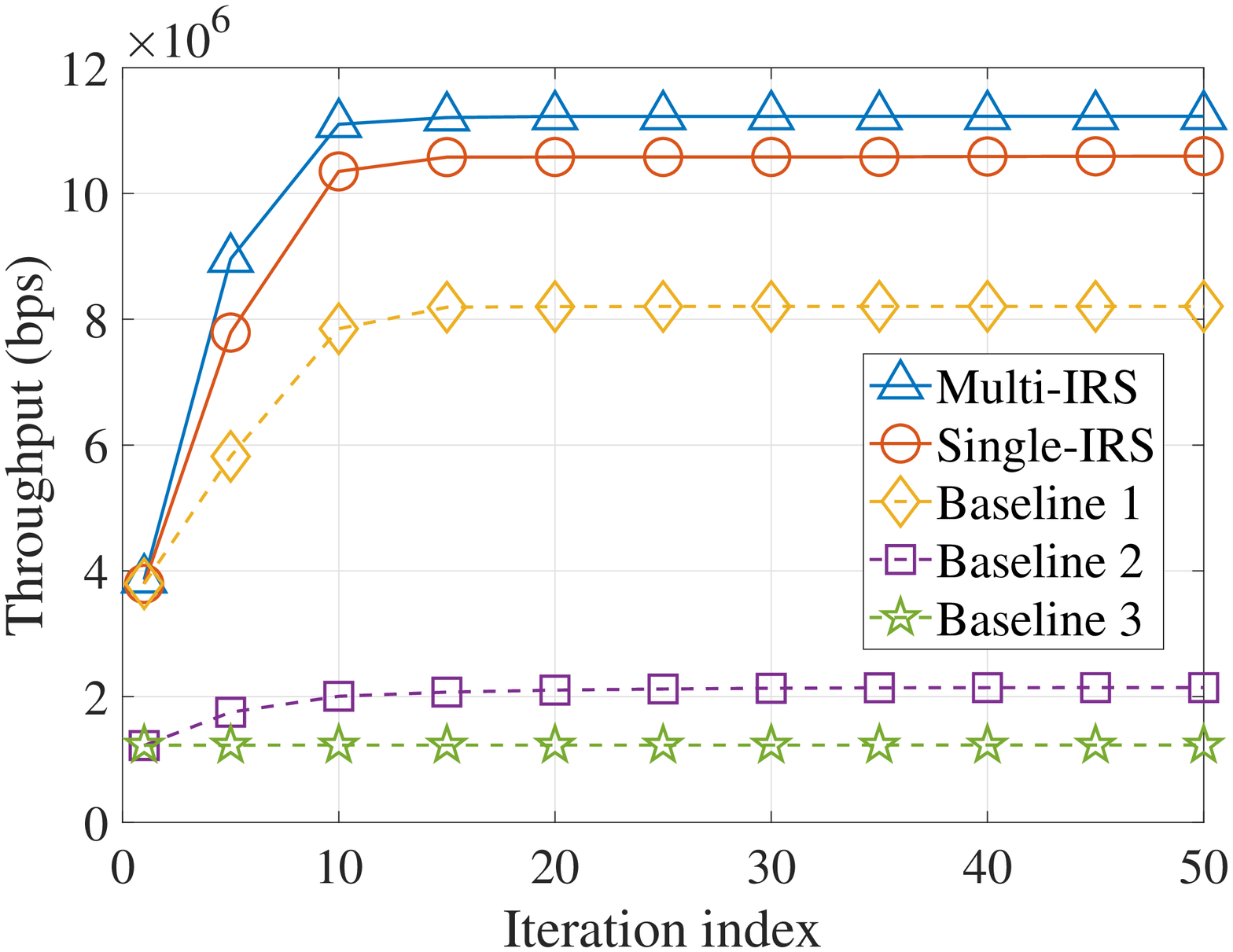}
\vspace{-.2cm}
\caption{Throughput v.s. number of iterations.}
\label{fig_sim}
\end{minipage}
\begin{minipage}[t]{0.33\textwidth}
\includegraphics[scale = 0.25]{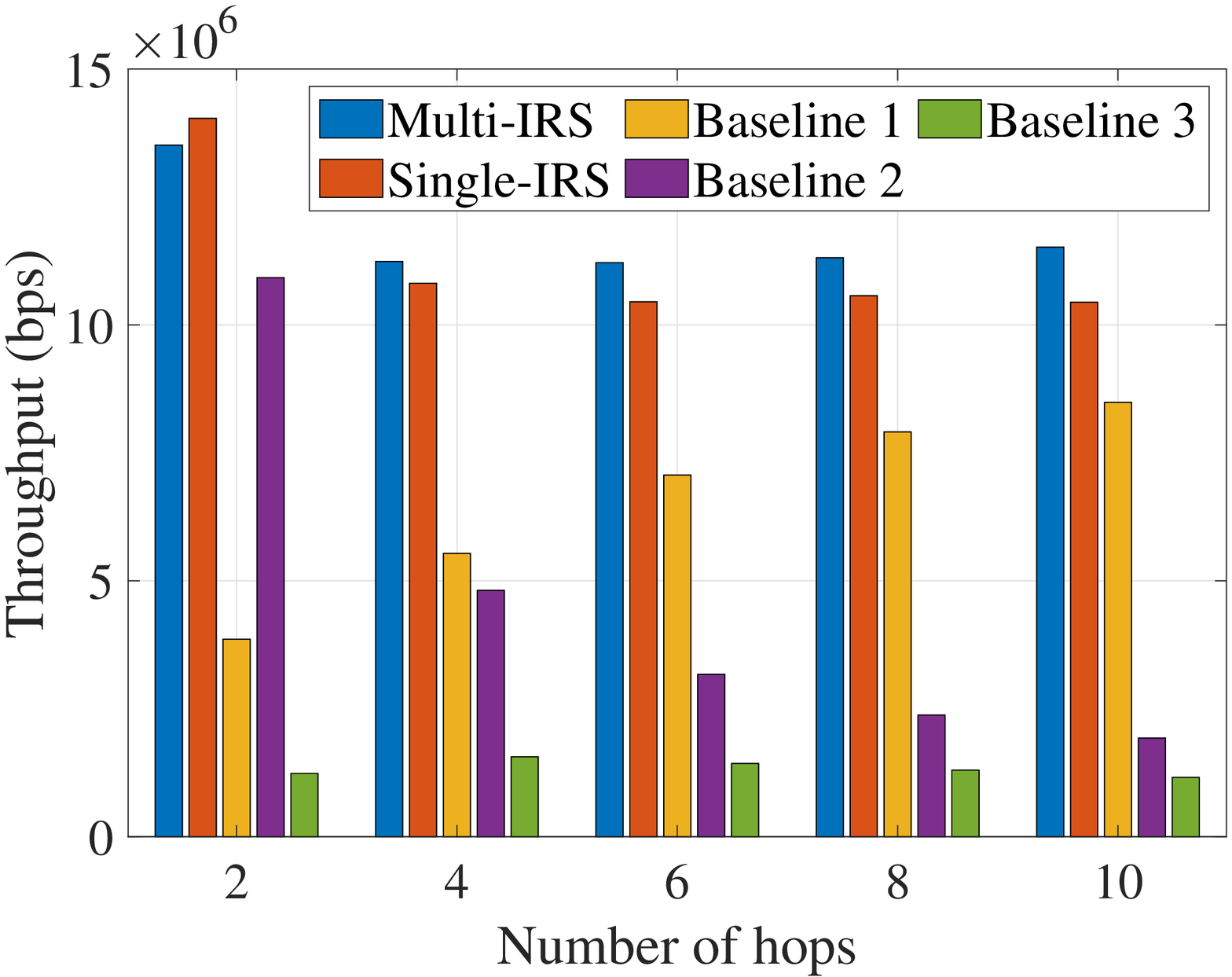}
\vspace{-.2cm}
\caption{Throughput v.s. number of hops.}
\label{fig_sim}
\end{minipage}
\begin{minipage}[t]{0.33\textwidth}
\includegraphics[scale = 0.25]{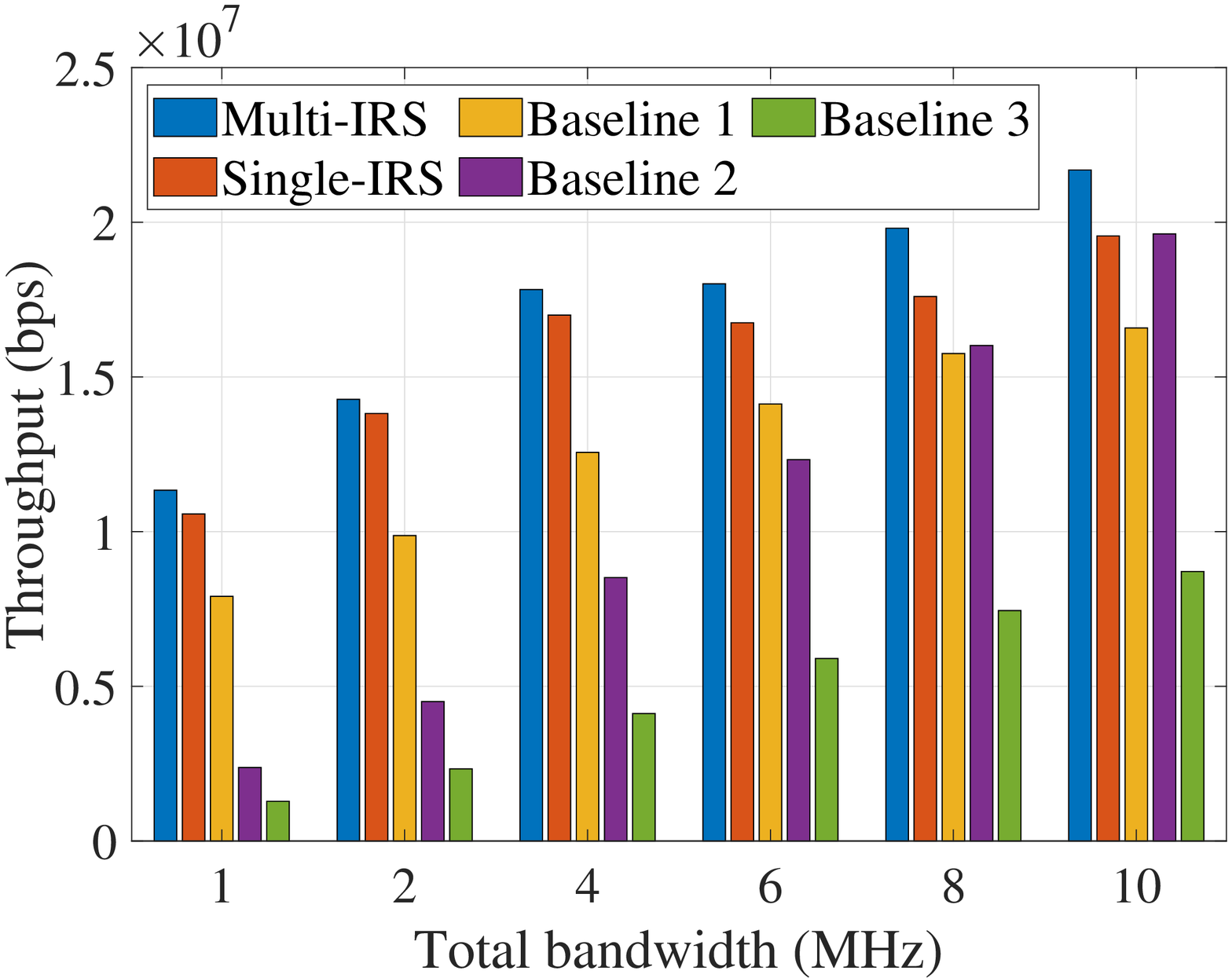}
\vspace{-.2cm}
\caption{Throughput v.s. total bandwidth.}
\label{fig_sim}
\end{minipage}
\end{figure*}

\section{Simulation Results}

\noindent In this section, numerical experiments are conducted to corroborate the effectiveness of the proposed scheme. 

The simulation parameters are as follows. The s-d distance is set to $L=1000$ (m), the distance from the s-d line to the y-z plane is set to $D=50$ (m). For the MDs, the total power for local computing and task offloading is set to $P_i=1$ (W). The altitude of the reference element of the IRS is set to $h_r=50$ (m), and the total number of reflection elements of the IRS is set to $M=10000$, with $M_z=M_y=100$. Since the IRS is assumed to be deployed on a favorable altitude such that the channels between the IRS and any of the MDs are dominated by LoS paths, the attenuation factor of these LoS channels is set to $\alpha_1=2$. The attenuation factor for the Rician fading channels, i.e., the channels among the MDs, is set to $\alpha_2=3.8$, since there are usually masses of obstacles and scatters on the ground. Other parameters are set as $\rho = -30$ (dB), $\beta = 3$ (dB), $\varepsilon=700$, and $\kappa=1 \times 10^{-28}$. The number of nodes and the total system bandwidth are set to $N=10$ and $B=1$ (MHz), respectively, unless otherwise stated. For each scheme, 1000 Monte-Carlo runs are conducted.

In the simulations, both the single- and multi-IRS cases are considered with identical number of reflection elements. Specifically, for the single-IRS case, the y-coordinate of the reference element of the IRS is set to $y_r=L/2$. For the multi-IRS case, it is assumed that there are $K=10$ IRSs, with each being a $50\times20$ UPA. These IRSs are assumed to be uniformly located above the projection of the s-d line in the y-axis. Specifically, for the $k$-th ($k=1,2,...,K$) IRS, the y-coordinate of its reference element is $y_r^k=(k-1)L/(K-1)$. In addition, to demonstrate the advantages of the proposed schemes, three baseline schemes are considered for comparison. Specifically, in baseline 1, the power allocation of the MDs and the system bandwidth allocation are jointly optimized without the aid of IRS. In baseline 2, no local computing is performed and the network nodes form a conventional multi-hop relay network assisted by a single-IRS. In baseline 3, the network nodes form a conventional multi-hop relay network without IRS.

First, the throughput performance of the proposed scheme is compared with that of the baseline schemes in  Fig. 3. It is observed that the proposed scheme can achieve a larger throughput as compared to the baseline schemes. The reason is that proper phase-shifts of the IRS and proper bandwidth allocation can improve the transmisnion rate for task offloding, while proper power allocation of the MDs can help strike a balance between local computing and task offloading. Specifically, the proposed scheme can achieve a throughput of about $11.2\times10^6$ (bps) for the multi-IRS case, and about $10.5\times10^6$ (bps) for the single-IRS case. Since there is no IRS deployed in baseline 1, the transmission rate for task offloading is not satisfactory, and hence it can only achieve a throughput of about $8.2\times10^6$ (bps). For baseline 2, although an IRS is deployed to improve the qualities of the transmission links, it can achieve a throughput of only about $2.1\times10^6$ (bps), as local computing is not performed in beseline 2. Baseline 3 gives the worst performance of about $1.2\times10^6$ (bps), since neither local computing nor IRS is considered.

The impact of the number of hops on the network throughput is shown in Fig. 4. Specifically, all parameters are kept unchanged, except for the number of hops. As it can be seen from Fig. 4, when the number of hops is $2$, the single-IRS scheme can achieve the largest throughput, bacause the single-IRS is located above the middle point of the projection of the s-d line in the y-axis, resulting in the shortest distance between the IRS and the relay node. The throughput achieved by the multi-IRS scheme is slightly smaller than that of the single-IRS scheme, because these multiple IRSs are deployed uniformly in a line, leading to relatively longer distances between the IRSs and the relay node. In addition, the achievable throughput of baseline 1 increases in the number of hops. This is because of that the number of nodes performing local computing is positively correlated with the number of hops. In contrast, the achievable throughput of baseline 2 decreases as the number of hops increases, as the nodes exhaust all their available power for task offloading in this scheme. In this case, the network throughput hinges on the transmission rates of all of the hops, which however decreases when more relays share the total bandwidth. For baseline 3, when the number of hops is small, the bandwidth allocated to each hop is relative abundant but the distance of each hop is relatively long. As the number of hops increases, less bandwidth is available to each hop but the distance of each hop becomes shorter. As a result, the throughput of baseline 3 remains unchanged at a relatively low value. In contrast, the proposed scheme can achieve substantially higher throughputs than that of the baselines, by concertedly exploiting the advantages of both local computing and IRS-assisted transmission.

Finally, the impact of the total bandwidth on the network throughput is investigated in the setting of an 8-hop network. As shown in Fig. 5, for all the schemes, the network throughputs increase as the total bandwidth increases. In particular, when the total bandwidth is sufficiently large, e.g., $B=10$ (MHz), the performance of the proposed scheme and that of baseline 2 are quite close. The reason is that the achievable throughput of the proposed scheme (single-IRS case) is dominated by the task offloading under this circumstance. However, when the available total bandwidth is limited, the throughputs of the proposed schemes are significantly larger than that of the baseline schemes.

\section{Conclusion}
\noindent In this paper, throughput optimization of an IRS-assisted multi-hop MEC network is investigated. Due to the coupling among the transmission links of different hops and the complicated multi-hop network topology, it is difficult to derive a closed-form expression of the network throughput. To tackle this challenge, the original problem is transformed into a max-flow problem in directed graph. By exploiting the special structure of the directed graph and the results from spectral graph theory, a joint phase-shifts, power allocation, and bandwidth allocation optimization algorithm is proposed to obtain a high-quality solution. Numerical results show that the proposed algorithm can achieve a substantially higher network throughput as compared to the baselines.

\bibliographystyle{IEEEtran}                             
\bibliography{IEEEabrv,references}

\begin{thebibliography}{10}
\providecommand{\url}[1]{#1}
\csname url@samestyle\endcsname
\providecommand{\newblock}{\relax}
\providecommand{\bibinfo}[2]{#2}
\providecommand{\BIBentrySTDinterwordspacing}{\spaceskip=0pt\relax}
\providecommand{\BIBentryALTinterwordstretchfactor}{4}
\providecommand{\BIBentryALTinterwordspacing}{\spaceskip=\fontdimen2\font plus
\BIBentryALTinterwordstretchfactor\fontdimen3\font minus
  \fontdimen4\font\relax}
\providecommand{\BIBforeignlanguage}[2]{{%
\expandafter\ifx\csname l@#1\endcsname\relax
\typeout{** WARNING: IEEEtran.bst: No hyphenation pattern has been}%
\typeout{** loaded for the language `#1'. Using the pattern for}%
\typeout{** the default language instead.}%
\else
\language=\csname l@#1\endcsname
\fi
#2}}
\providecommand{\BIBdecl}{\relax}
\BIBdecl

\bibitem{mao2017survey}
Y.~Mao, C.~You, J.~Zhang, K.~Huang, and K.~B. Letaief, ``A survey on mobile
  edge computing: The communication perspective,'' \emph{IEEE Commun. Surv.
  Tut.}, vol.~19, no.~4, pp. 2322--2358, Aug. 2017.

\bibitem{he2019peace}
X.~He, R.~Jin, and H.~Dai, ``Peace: Privacy-preserving and cost-efficient task
  offloading for mobile-edge computing,'' \emph{IEEE Trans. Wireless Commun.},
  vol.~19, no.~3, pp. 1814--1824, Dec. 2019.

\bibitem{wang2016mobile}
Y.~Wang, M.~Sheng, X.~Wang, L.~Wang, and J.~Li, ``Mobile-edge computing:
  Partial computation offloading using dynamic voltage scaling,'' \emph{IEEE
  Trans. Commun.}, vol.~64, no.~10, pp. 4268--4282, Aug. 2016.

\bibitem{mahmoodi2019optimal}
S.~E. Mahmoodi, R.~Uma, and K.~Subbalakshmi, ``Optimal joint scheduling and
  cloud offloading for mobile applications,'' \emph{IEEE Trans. Cloud Comput.},
  vol.~7, no.~2, pp. 301--313, Apr. 2019.

\bibitem{wu2019towards}
Q.~Wu and R.~Zhang, ``Towards smart and reconfigurable environment: Intelligent
  reflecting surface aided wireless network,'' \emph{IEEE Commun. Mag.},
  vol.~58, no.~1, pp. 106--112, Nov. 2019.

\bibitem{wu2019intelligent}
------, ``Intelligent reflecting surface enhanced wireless network via joint
  active and passive beamforming,'' \emph{IEEE Trans. Wireless Commun.},
  vol.~18, no.~11, pp. 5394--5409, Aug. 2019.

\bibitem{wu2019beamforming}
------, ``Beamforming optimization for wireless network aided by intelligent
  reflecting surface with discrete phase shifts,'' \emph{IEEE Trans. Commun.},
  vol.~68, no.~3, pp. 1838--1851, Dec. 2019.

\bibitem{zhou2020robust}
G.~Zhou, C.~Pan, H.~Ren, K.~Wang, M.~Di~Renzo, and A.~Nallanathan, ``Robust
  beamforming design for intelligent reflecting surface aided {MISO}
  communication systems,'' \emph{IEEE Wireless Commun. Lett.}, vol.~9, no.~10,
  pp. 1658--1662, Jun. 2020.

\bibitem{di2020hybrid}
B.~Di, H.~Zhang, L.~Song, Y.~Li, Z.~Han, and H.~V. Poor, ``Hybrid beamforming
  for reconfigurable intelligent surface based multi-user communications:
  Achievable rates with limited discrete phase shifts,'' \emph{IEEE J. Sel.
  Areas Commun.}, vol.~38, no.~8, pp. 1809--1822, Jun. 2020.

\bibitem{li2020reconfigurable}
S.~Li, B.~Duo, X.~Yuan, Y.-C. Liang, and M.~Di~Renzo, ``Reconfigurable
  intelligent surface assisted {UAV} communication: Joint trajectory design and
  passive beamforming,'' \emph{IEEE Wireless Commun. Lett.}, vol.~9, no.~5, pp.
  716--720, Jan. 2020.

\bibitem{bai2020latency}
T.~Bai, C.~Pan, Y.~Deng, M.~Elkashlan, A.~Nallanathan, and L.~Hanzo, ``Latency
  minimization for intelligent reflecting surface aided mobile edge
  computing,'' \emph{IEEE J. Sel. Areas Commun.}, vol.~38, no.~11, pp.
  2666--2682, Jul. 2020.

\bibitem{liu2020intelligent}
Y.~Liu, J.~Zhao, Z.~Xiong, D.~Niyato, Y.~Chau, C.~Pan, and B.~Huang,
  ``Intelligent reflecting surface meets mobile edge computing: Enhancing
  wireless communications for computation offloading,'' \emph{arXiv preprint
  arXiv:2001.07449}, 2020.

\bibitem{bai2020resource}
T.~Bai, C.~Pan, H.~Ren, Y.~Deng, M.~Elkashlan, and A.~Nallanathan, ``Resource
  allocation for intelligent reflecting surface aided wireless powered mobile
  edge computing in {OFDM} systems,'' \emph{arXiv preprint arXiv:2003.05511},
  2020.

\bibitem{cao2019intelligent}
Y.~Cao and T.~Lv, ``Intelligent reflecting surface enhanced resilient design
  for {MEC} offloading over millimeter wave links,'' \emph{arXiv preprint
  arXiv:1912.06361}, 2019.

\bibitem{chung1997spectral}
F.~R. Chung and F.~C. Graham, \emph{Spectral graph theory}.\hskip 1em plus
  0.5em minus 0.4em\relax Providence, RI, USA: Amer. Math. Soc., 1997.

\bibitem{ford1956maximal}
L.~R. Ford and D.~R. Fulkerson, ``Maximal flow through a network,'' \emph{Can.
  J. Math.}, vol.~8, pp. 399--404, 1956.

\bibitem{bhattacharya2010graph}
S.~Bhattacharya and T.~Ba{\c{s}}ar, ``Graph-theoretic approach for connectivity
  maintenance in mobile networks in the presence of a jammer,'' in \emph{Proc.
  IEEE Conf. Decision Control (CDC)}, Atlanta, Georgia, Dec. 2010.

\bibitem{he2014dynamic}
X.~He, H.~Dai, and P.~Ning, ``Dynamic adaptive anti-jamming via controlled
  mobility,'' \emph{IEEE Trans. Wireless Commun.}, vol.~13, no.~8, pp.
  4374--4388, Apr. 2014.

\bibitem{ju2012full}
H.~Ju, S.~Lim, D.~Kim, H.~V. Poor, and D.~Hong, ``Full duplexity in
  beamforming-based multi-hop relay networks,'' \emph{IEEE J. Sel. Areas
  Commun.}, vol.~30, no.~8, pp. 1554--1565, Aug. 2012.

\bibitem{zhang2013energy}
W.~Zhang, Y.~Wen, K.~Guan, D.~Kilper, H.~Luo, and D.~O. Wu, ``Energy-optimal
  mobile cloud computing under stochastic wireless channel,'' \emph{IEEE Trans.
  Wireless Commun.}, vol.~12, no.~9, pp. 4569--4581, Aug. 2013.

\bibitem{zhang2017energy}
J.~Zhang, X.~Hu, Z.~Ning, E.~C.-H. Ngai, L.~Zhou, J.~Wei, J.~Cheng, and B.~Hu,
  ``Energy-latency tradeoff for energy-aware offloading in mobile edge
  computing networks,'' \emph{IEEE Internet Things J.}, vol.~5, no.~4, pp.
  2633--2645, Dec. 2017.

\bibitem{wang2017computation}
C.~Wang, C.~Liang, F.~R. Yu, Q.~Chen, and L.~Tang, ``Computation offloading and
  resource allocation in wireless cellular networks with mobile edge
  computing,'' \emph{IEEE Trans. Wireless Commun.}, vol.~16, no.~8, pp.
  4924--4938, May 2017.

\bibitem{theproof}
\BIBentryALTinterwordspacing
``Proof.'' [Online]. Available:
  \url{https://1drv.ms/b/s!Avz\_DER2wzmJaxhr0ra3JwArSWs}
\BIBentrySTDinterwordspacing

\end{thebibliography}

\end{document}